\documentclass[11pt]{article}
\usepackage{amsmath,amssymb,amsthm}

\newcommand{\rd}{\partial}

\newcommand{\CC}{\mathbf{C}}

\newcommand{\ZZ}{\mathbf{Z}}
\newcommand{\CP}{\mathbf{CP}}

\newcommand{\bst}{\boldsymbol{t}}

\newcommand{\calL}{\mathcal{L}}
\newcommand{\calM}{\mathcal{M}}
\newcommand{\calS}{\mathcal{S}}

\newcommand{\calZ}{\mathcal{Z}}
\newtheorem{theorem}{Theorem}[section]
\newtheorem{proposition}[theorem]{Proposition}

\theoremstyle{definition}

\theoremstyle{remark} \newtheorem{remark}[theorem]{Remark}

\newcommand{\pa}{\partial}

\DeclareMathOperator{\ad}{ad}
\renewcommand{\setminus}{\smallsetminus}
\begin{document}

\title{Non-degenerate solutions\\
of universal Whitham hierarchy}
\author{Kanehisa Takasaki\\
Graduate School of Human and Environmental Studies,\\
Kyoto University,\\
Yoshida, Sakyo, Kyoto, 606-8501, Japan
\\ \\
Takashi Takebe\\
Faculty of Mathematics,\\
State University -- Higher School of Economics,\\
Vavilova Street, 7, Moscow, 117312, Russia
\\ \\
Lee Peng Teo\\
Department of Applied Mathematics, Faculty of Engineering, \\
University of Nottingham Malaysia Campus, Jalan Broga,\\ 
43500, Semenyih, Selangor Darul Ehsan, Malaysia.   }
\date{}
\maketitle

\begin{abstract}
The notion of non-degenerate solutions for the dispersionless Toda 
hierarchy is generalized to the universal Whitham hierarchy of genus 
zero with $M+1$ marked points.  These solutions are characterized by a 
Riemann-Hilbert problem (generalized string equations) with respect to 
two-dimensional canonical transformations, and may be thought of as a 
kind of general solutions of the hierarchy.  The Riemann-Hilbert problem 
contains $M$ arbitrary functions $H_a(z_0,z_a)$, $a = 1,\ldots,M$, which 
play the role of generating functions of two-dimensional canonical 
transformations.  The solution of the Riemann-Hilbert problem is 
described by period maps on the space of $(M+1)$-tuples $(z_\alpha(p) : 
\alpha = 0,1,\ldots,M)$ of conformal maps from $M$ disks of the Riemann 
sphere and their complements to the Riemann sphere.  The period maps are 
defined by an infinite number of contour integrals that generalize the 
notion of harmonic moments.  The $F$-function (free energy) of these 
solutions is also shown to have a contour integral representation. 
\end{abstract}

\newpage 

\section{Introduction}

The universal Whitham hierarchy is a unified framework for various
dispersionless integrable systems and Whitham modulation equations
\cite{Krichever94}.  In particular, the hierarchy of genus zero, which
is the subject of this paper, is a natural generalization of the
dispersionless KP and Toda hierarchies \cite{TT95}.  Therefore it is
natural to ask to what extent the rich contents of the dispersionless KP
and Toda hierarchies can be generalized to the hierarchy of genus zero.

This issue has been sought for since the turn of the century
when the study of dispersionlss integrable systems entered
a new stage.  As regards the problem of special solutions,
the classical ``hodograph method'' has been generalized
\cite{GMMA03,TT08} to obtain a class of solutions
including Krichever's ``algebraic orbits'' \cite{Krichever94}.
Another class of special solutions (also related to algebraic orbits)
have been studied in the context of the Virasoro constraints
\cite{MMMA05,MAMM05} and the large-$N$ limit of multiple
orthogonal polynomials \cite{MAM08}.  It should be stressed
that the structure of infinitesimal additional symmetries
(including the Virasoro symmetries) was fully elucidated
by the work of the Madrid group \cite{MMMA05,MAMM05}.
As in the case of the dispersionless KP and Toda hierarchies
\cite{TT95}, those symmetries are derived from a ``nonlinear''
Riemann-Hilbert problem (or an equivalent $\bar{\rd}$ problem
\cite{KMA01a,KMA01b}) with respect to two-dimensional
canonical transformations.
As regards the Riemann-Hilbert problem itself, however,
no effective method for finding an explicit form of solutions
is known apart from very special cases;  one has to resort
to a genuine existence theorem (though it is enough for
deriving the infinitesimal symmetries).   Moreover,
the $F$-function (free energy), also known as
the dispersionless (logarithm of) tau function,
has to be treated separately in this approach.

Recently, one of the present authors reformulated the Riemann-Hilbert
problem for the dispersionless Toda hierarchy in a slightly different
form, and introduced the notion of ``non-degenerate solutions'' for
which a more effective description is available \cite{Teo0906}.  A
central idea of this result stems from the work of Wiegmann and Zabrodin
\cite{WZ00} on an integrable structure of univalent conformal maps in
Riemann's mapping theorem.  They used the harmonic moments of the domain
to interpret the conformal maps as a special solution of the
dispersionless Toda hierarchy.  This result can be generalized to pairs
of conformal maps \cite{Teo0905}.  The harmonic moments are redefined
therein as contour integrals that include the conformal map (or the pair
of conformal maps), and shown to give a system of local coordinates on
the space of pairs of conformal maps.  Actually, this amounts to solving
a Riemann-Hilbert problem (or ``string equations'') in a special case
\cite{Takasaki95}.  The method of harmonic moments were generalized
later by Zabrodin to a larger class of solutions of the dispersionless
Toda hierarchy \cite{Zabrodin01}.  The notion of non-degenerate
solutions is a rigorous reformulation of those solutions, which thereby
turn out to be a kind of general (or generic) solutions of the
dispersionless Toda hierarchy.  The goal of this paper is to generalize
these results \cite{Teo0906} to the universal Whitham hierarchy of genus
zero.

Let us briefly recall the notion of non-degenerate solutions
of the dispersionless Toda hierarchy.  Those solutions
are characterized by a Riemann-Hilbert problem
of the following form:   Let $H(z,\tilde{z})$
be a holomorphic function of two variables
defined in a suitable domain (not specified here),
and $H_z(z,\tilde{z})$ and $H_{\tilde{z}}(z,\tilde{z})$
denote the derivatives
$H_z(z,\tilde{z}) = \rd H(z,\tilde{z})/\rd z$,
$H_{\tilde{z}}(z,\tilde{z}) = \rd H(z,\tilde{z})/\rd\tilde{z}$.
Moreover, suppose that $H(z,\tilde{z})$ satisfies
the non-degeneracy condition
\begin{eqnarray*}
  H_{z\tilde{z}}(z,\tilde{z}) \not= 0.
\end{eqnarray*}
The problem is to find four functions
$\calL(P),\calM(P),\tilde{\calL}(P),\tilde{\calM}(P)$
of a complex variable $P$ with the following properties:
\begin{itemize}
\item[(i)] $\calL(P)$ and $\calM(P)$ are holomorphic functions
in the punctured disk $1 < |P| < \infty$,
$\calL(P)$ being univalent therein,
and have a Laurent expansion of the form
\begin{eqnarray*}
\begin{aligned}
  \calL(P) &= P + \sum_{n=1}^\infty u_nP^{-n+1},\\
  \calM(P) &= \sum_{n=1}^\infty nt_n\calL(P)^n
    + t_0 + \sum_{n=1}^\infty v_n\calL(P)^n.
\end{aligned}
\end{eqnarray*}
\item[(ii)] $\tilde{\calL}(P)^{-1}$ and $\tilde{\calM}(P)$
are holomorphic functions in the punctured disk $0 < |P| < 1$,
$\tilde{\calL}(P)$ being univalent therein, and
have a Laurent expansion of the form
\begin{eqnarray*}
\begin{aligned}
  \tilde{\calL}(P)^{-1} &= \sum_{n=0}^\infty \tilde{u}_nP^{n-1}
  \quad (\tilde{u}_0 \not= 0),\\
  \tilde{\calM}(P) &= - \sum_{n=1}^\infty nt_{-n}\tilde{\calL}(P)^{-n}
    + t_0 - \sum_{n=1}^\infty v_{-n}\tilde{\calL}(P)^n.
\end{aligned}
\end{eqnarray*}
\item[(iii)] These functions can be analytically continued to a
neighborhood of the unit circle $|P| = 1$ and satisfy the functional
equations (generalized string equations)
\begin{eqnarray}
  \calM(P) = \calL(P) H_{z}(\calL(P),\tilde{\calL}(P)),\quad
  \tilde{\calM}(P) = - \tilde{\calL}(P)H_{\tilde{z}}(\calL(P),\tilde{\calL}(P))
\label{dToda-RH1}
\end{eqnarray}
therein.
\end{itemize}
If the equations
\begin{eqnarray*}
  w = zH_z(z,\tilde{z}),\quad
  \tilde{w} = - \tilde{z}H_{\tilde{z}}(z,\tilde{z})
\end{eqnarray*}
can be solved for $\tilde{z}$, the map
$(z,w) \mapsto (f(z,w),g(z,w)) = (\tilde{z},\tilde{w})$
becomes a two-dimensional canonical transformation
(or symplectic map) with respect to the symplectic form
\begin{eqnarray*}
  \frac{dz\wedge dw}{z}
  = \frac{d\tilde{z}\wedge d\tilde{w}}{\tilde{z}},
\end{eqnarray*}
the function $H(z,\tilde{z})$ being its ``generating function''.
It is well known that this is a normal form
of canonical transformations in a ``general position''
of the set of all canonical transformations.
(\ref{dToda-RH1}) can be thus rewritten as
\begin{eqnarray}
  \tilde{\calL}(P) = f(\calL(P),\calM(P)),\quad
  \tilde{\calM}(P) = g(\calL(P),\calM(P)).
\label{dToda-RH2}
\end{eqnarray}
This is a Riemann-Hilbert problem of the standard form
that characterizes the Lax and Orlov-Schulman functions
of the dispersionless Toda hierarchy \cite{TT95}.
The aforementioned remark on canonical transformations
with generating functions imply that
the non-degenerate solutions are indeed
general solutions of the dispersionless Toda hierarchy.

An advantage of (\ref{dToda-RH1}) over (\ref{dToda-RH2})
is that it is ``solvable'' in the following sense.
The generalized string equations (\ref{dToda-RH1})
can be converted to the infinite system of equations
\begin{eqnarray}
\begin{aligned}
  nt_n &= \frac{1}{2\pi i}\oint_{|P|=1}
        H_z(\calL(P),\tilde{\calL}(P))\calL(P)^{-n}d\calL(P),\\
  nt_{-n} &= \frac{1}{2\pi i}\oint_{|P|=1}
        H_{\tilde{z}}(\calL(P),\tilde{\calL}(P))\tilde{\calL}(P)^n
        d\tilde{\calL}(P),\\
  t_0 &= \frac{1}{2\pi i}\oint_{|P|=1}
        H_z(\calL(P),\tilde{\calL}(P))d\calL(P)\\
      &= - \frac{1}{2\pi i}\oint_{|P|=1}
        H_{\tilde{z}}(\calL(P),\tilde{\calL}(P))d\tilde{\calL}(P),
\end{aligned}
\label{dToda-tn-oint}
\end{eqnarray}
and
\begin{eqnarray}
\begin{aligned}
  v_n &= \frac{1}{2\pi i}\oint_{|P|=1}
        H_z(\calL(P),\tilde{\calL}(P))\calL(P)^nd\calL(P),\\
  v_{-n} &= \frac{1}{2\pi i}\oint_{|P|=1}
        H_{\tilde{z}}(\calL(P),\tilde{\calL}(P))\tilde{\calL}(P)^{-n}
        d\tilde{\calL}(P)
\end{aligned}
\label{dToda-vn-oint}
\end{eqnarray}
for $n = 1,2,\ldots$.  Note that the contour integrals
are analogues of harmonic moments; in the terminology
of geometry, they are a kind of ``period integrals''.
A fundamental fact \cite{Teo0906} is that the first set
(\ref{dToda-tn-oint}) of these period integrals
give a system of local coordinates on the space
of the pairs $(\calL,\tilde{\calL})$ of conformal maps.
This implies that the ``period map''
$(\calL,\tilde{\calL}) \mapsto (t_n : n \in \ZZ)$
is (locally) invertible, and the inverse map
and the second set (\ref{dToda-vn-oint}) of period integrals
give a (unique) solution of the Riemann-Hilbert problem.
Remarkably, the $F$-function, too, turns out to have
a contour integral representation \cite{Teo0906}.

In the language of the universal Whitham hierarchy
of genus zero, the dispersionless Toda hierarchy
amounts to the case with two ``marked points''.
The general $(M+1)$-point hierarchy is formulated
by $M+1$ pairs $(z_\alpha(p),\zeta_\alpha(p))$,
$\alpha = 0,1,\ldots,M$, of Lax and Orlov-Schulman
functions.  In the two-point ($M = 1$) case,
these functions are connected with
the Lax and Orlov-Schulman functions
of the dispersionless Toda hierarchy as
\begin{eqnarray*}
\begin{aligned}
  &z_0(p) = \calL(P),& \quad
  &z_1(p) = \tilde{\calL}(P)^{-1},& \\
  &\zeta_0(p) = \calM(P)\calL(P)^{-1},&\quad
  &\zeta_1(p) = - \tilde{\calM}(P)\tilde{\calL}(P),&
\end{aligned}
\end{eqnarray*}
where the coordinates $p$ and $P$ of the Riemann sphere
in both hierarchies are related as
\begin{eqnarray*}
  p = P + u_1.
\end{eqnarray*}
Thus the marked points $P = \infty,0$ of
the dispersionless Toda hierarchy correspond
to the marked points $p = \infty,u_1$ of
the universal Whitham hierarchy.
Bearing this interpretation of the dispersionless Toda hierarchy
in mind, we turn to the $M+1$-point case.

This paper is organized as follows.  In Section 2, we review the
fundamental structure of the universal Whitham hierarchy of genus zero.
Building blocks of the hierarchy, such as the Lax and Orlov-Schulmann
functions, the $S$-functions, the $F$-function and the generalized
Grunsky coefficients, are introduced in detail.  For technical reasons,
the definition of the $F$-function in our previous work \cite{TT06,TT08}
is slightly modified here, though this is not a serious problem.  In
Section 3, we formulate the Riemann-Hilbert problem that defines
non-degenerate solutions.  The basic setup is parallel to the
formulation by the Madrid group \cite{MMMA05,MAMM05}.  Our generalized
string equations have $M$ arbitrary functions $H_a(z_0,z_a)$, $a =
1,\ldots,M$, as functional data.  As in the case of the dispersionless
Toda hierarchy, these functions play the role of generating functions of
two-dimensional canonical transformations.  In Section 4, we generalize
the period integrals (\ref{dToda-tn-oint}) and (\ref{dToda-vn-oint}) to
the space $\calZ$ of $(M+1)$-tuples $(z_\alpha(p) : \alpha =
0,1,\ldots,M)$ of conformal maps, and show that a half of them give a
system of local coordinates on $\calZ$.  This justifies the definition
of non-degenerate solutions.  Section 5 is an intermediate step towards
the construction of the $F$-function.  We present here a contour
integral representation of the potentials $\phi_a, a=1,\ldots,M,$ that
show up in the Laurent expansions of the $S$-functions.  These
$\phi$-functions are used in Section 6 for the construction of the
$F$-function.  As in the case of the dispersionless Toda hierarchy, we
define a set of auxiliary functions $J_{a,1}(z_0,z_a),J_{a,2}(z_0,z_a)$,
$a = 1,\ldots,M$.  These functions are used to express the $F$-function
in terms of contour integrals.  In Section 7, we illustrate the
construction of non-degenerate solutions in a few special cases that
amount to the examples studied for the dispersionless Toda hierarchy
\cite{Teo0906}.

\paragraph*{\em Acknowledgements}
This work is partly supported by Grant-in-Aid for Scientific Research
No.\ 19104002, \ 19540179 and No.\ 21540218 
from the Japan Society for the Promotionof Science.  
TT is partly supported by the grant of the State University
-- Higher School of Economics, Russia, for the Individual Research
Project 09-01-0047 (2009).

\section{Building blocks of universal Whitham hierarchy}

In this section we review essential facts on the universal Whitham
hierarchy of genus zero necessary for our later discussion, following
our previous work \cite{TT06,TT08}\footnote{The authors of \cite{TT06}
sincerely apologize numerous typographical errors in the proofs in it,
but the statements there are correct. The only differences from
\cite{TT06} are the definition of the $F$-function \eqref{eq2_22_5} and,
consequently, changes of several signatures in, e.g., \eqref{Grunsky}.}.
The notations are mostly the same as 
\cite{TT06,TT08}, except that, after the notation
of the recent work \cite{MAM08} of the Madrid group,
Greek indices $\alpha,\beta,\ldots$ range over $0,1,\ldots,M$
and Latin indices $a,b,\ldots$ over $1,\ldots,M$.

\paragraph{Lax functions}
The Lax functions $z_\alpha(p)$, $\alpha = 0,1,\ldots,M$,
are functions with Laurent expansions of the form
\begin{eqnarray}
\begin{aligned}
  z_0(p) &= p + \sum_{j=2}^\infty u_{0j}p^{-j+1},\\
  z_a(p) &= \frac{r_a}{p-q_a} + \sum_{j=1}^\infty u_{aj}(p-q_a)^{j-1}
  \quad (a = 1,\ldots,M),
\end{aligned}
\label{z-expansion}
\end{eqnarray}
in a neighborhood of $p = \infty$ and $p = q_a$, respectively.
The coefficients $u_{\alpha j}$
($r_a = u_{a0}$) and the centers $q_a$ are dynamical variables.
To consider a Riemann-Hilbert problem \cite{MMMA05,MAMM05},
we choose a set of disjoint positively oriented simple closed curves
$C_1,\ldots,C_M$ that encircle $q_1,\ldots,q_M$ counterclockwise,
and assume that the Laurent expansion of $z_a(p)$ converges 
in the inside $D_a$ of $C_a$ and that the Laurent expansion 
of $z_0(p)$ converges in a neighborhood of $p = \infty$ 
and can be analytically continued, as a holomorphic function, 
to the outside $\CC\setminus(D_1\cup\cdots\cup D_M)$ of $D_a$'s.

\paragraph{Lax equations}
The hierarchy has $M+1$ series of time evolutions
with time variables $t_{0n}$, $n = 1,2,\ldots$
and $t_{an}$, $a = 1,\ldots,M$, $n = 0,1,2,\ldots$.
The time evolutions of the Lax functions
are defined by the Lax equations
\begin{eqnarray}
  \rd_{\alpha n}z_\beta(p)
  = \{\Omega_{\alpha n}(p),\, z_\beta(p)\}, \quad
  \rd_{\alpha n} = \rd/\rd t_{\alpha n},
\end{eqnarray}
with respect to the Poisson bracket
\begin{eqnarray}
  \{f,g\}
  = \frac{\rd f}{\rd p}\frac{\rd g}{\rd t_{01}}
  - \frac{\rd f}{\rd t_{01}}\frac{\rd g}{\rd p}.
\end{eqnarray}
The Hamiltonians $\Omega_{\alpha n}(p)$ are defined as
\begin{eqnarray}
\begin{aligned}
  &\Omega_{0n}(p) = \bigl(z_0(p)^n\bigr)_{(0,\ge 0)}, \quad
  \Omega_{a n}(p) = \bigl(z_a(p)^n\bigr)_{(a,<0)}
  \quad (n = 1,2,\ldots), \\
  &\Omega_{a 0}(p) = - \log(p - q_a),
\end{aligned}
\end{eqnarray}
where $(\quad)_{(0,\ge 0)}$ denotes the projection
to non-negative powers of $p$, and $(\quad)_{(a,<0)}$
the projection to negative powers of $p - q_a$.
In other words,
\begin{eqnarray}
\begin{aligned}
  z_0(p)^n &= \Omega_{0n}(p) + O(p^{-1})  \quad (p \to \infty), \\
  z_a(p)^n &= \Omega_{an}(p) + O(1)  \quad (p \to q_a)
\end{aligned}
\end{eqnarray}
for $n \ge 1$.  $\Omega_{\alpha n}(p)$ satisfies
the dispersionless Zakharov-Shabat equations
\begin{eqnarray}
  \rd_{\beta m}\Omega_{\alpha n}(p)
  - \rd_{\alpha n}\Omega_{\beta m}(p)
  + \{\Omega_{\alpha n}(p),\,\Omega_{\beta m}(p)\}
  = 0.
\end{eqnarray}

As pointed out in \cite{MMMA05}, the dressing functions of the universal
Whitham hierarchy have the following form:
\begin{equation}
\label{eq1_21_3}
\begin{split}
\varphi_0(p)=& \sum_{j=1}^{\infty}\varphi_{0,j}p^{-j},
\qquad
\varphi_{a}(p)=\sum_{j=0}^{\infty}\varphi_{a,j}(p-q_a^{(0)})^j,
\end{split}
\end{equation}
\begin{equation*}\begin{split}
z_0(p)=&e^{\text{ad}\varphi_0(p)}p,\qquad
z_{a}(p)=e^{\text{ad}\varphi_{a}(p)}(p-q_a^{(0)})^{-1}.
\end{split}
\end{equation*}
The following is due to \cite{MMMA05}, Theorem3.
\begin{proposition}
 \label{p1}
 If $(z_{\alpha}(p)\,:\,\alpha=0,1,\ldots,M)$ is a solution of the
 universal Whitham hierarchy, then there exists dressing functions
 $\varphi_{\alpha}(p)$ of the form \eqref{eq1_21_3}, such that
\begin{equation}
\label{eq1_21_4}
\begin{split}
 z_0(p)=&e^{\text{ad}\varphi_0(p)}p,\qquad
 z_{a}(p)=e^{\text{ad}\varphi_{a}(p)}(p-q_a^{(0)})^{-1},
\end{split}\end{equation}
and
\begin{equation}
   \nabla_{\alpha n} \varphi_{\beta}=
   \tilde{\Omega}_{\alpha n,\beta},
\end{equation}
where
\begin{equation}
   \tilde{\Omega}_{\alpha n,\beta}
   =\begin{cases}
    \Omega_{\alpha n}-\delta_{\alpha 0}\delta_{n1}z_{\beta}(p)^{-1}\quad
    &( \alpha\neq \beta\ \text{and}\ (\beta\neq 0\ \text{or}\ n\neq 0) ),
    \\
    \Omega_{\alpha 0}+\log z_0(p)\quad
    &(\alpha\neq 0\ \text{and}\ \beta=0\ \text{and}\ n=0),
    \\
    \Omega_{\alpha n}-z_{\alpha}(p)^n
    &(\alpha=\beta\ \text{and}\ n\neq 0),
    \\
    \Omega_{\alpha 0}-\log z_{\alpha }(p)
    &(\alpha=\beta\neq 0\ \text{and}\ n=0),
   \end{cases}
\label{diff-eq:dressing}
\end{equation}
and
$\nabla_{\alpha n}$ is the right logarithmic derivative (cf.\
\cite{MMMA05} Appendix A, \cite{TT95} Appendix A) defined by
\begin{equation}
    \nabla_{\alpha n} \psi
    =
    \sum_{n=0}^\infty
    \frac{(\ad \psi)^n}{(n+1)!} \partial_{\alpha n} \psi.
\label{def:delta}
\end{equation}
\end{proposition}
In the above, $q_a^{(0)}, a=1,\ldots, M,$ are arbitrary non-dynamical
variables. Without loss of generality, we set $q_a^{(0)}=0$ henceforth.

\paragraph{Orlov-Schulman functions}
The Orlov-Schulman functions $\zeta_\alpha(p)$, $\alpha = 0,1,\ldots,M$
are Laurent series of the form
\begin{eqnarray}
\begin{aligned}
  \zeta_0(p)
  &= \sum_{n=1}^\infty nt_{0n}z_0(p)^{n-1}
    + \frac{t_{00}}{z_0(p)}
    + \sum_{n=1}^\infty z_0(p)^{-n-1}v_{0n}, \\
  \zeta_a(p)
  &= \sum_{n=1}^\infty nt_{an}z_a(p)^{n-1}
    + \frac{t_{a0}}{z_a(p)}
    + \sum_{n=1}^\infty z_a(p)^{-n-1}v_{an},
\end{aligned}
\label{zeta-expansion}
\end{eqnarray}
where
\begin{eqnarray*}
  t_{00} = - \sum_{a=1}^M t_{a0}.
\end{eqnarray*}
They satisfy the Lax equations
\begin{eqnarray}
  \rd_{\alpha n}\zeta_\beta(p)
  = \{\Omega_{\alpha n}(p),\, \zeta_\beta(p)\}
\label{lax:zeta}
\end{eqnarray}
and the canonical Poisson commutation relation
\begin{eqnarray}
  \{z_\alpha(p),\zeta_\alpha(p)\} = 1.
\label{canonical(z,zeta)}
\end{eqnarray}
In terms of the dressing functions, $\zeta_{\alpha}$ are given by
\begin{equation*}
\begin{split}
\zeta_0(p)
=&e^{\text{ad}\varphi_0(p)}
\left(
   \sum_{n=1}^{\infty}n t_{0n} p^{n-1}+\frac{t_{00}}{p}
\right),
\\
\zeta_{a}(p)
=&e^{\text{ad}\varphi_a(p)}
\left(
   \sum_{n=1}^{\infty}n t_{a n} p^{-n+1}+ t_{a 0}p-t_{01}p^{2}
\right).
\end{split}
\end{equation*}
The canonical Poisson commutation relation \eqref{canonical(z,zeta)} is
a direct consequence of the definition and the Lax equations
\eqref{lax:zeta} follow from \eqref{diff-eq:dressing}.

\paragraph*{$S$-functions}
The $S$-functions $\calS_\alpha(p)$, $\alpha = 0,1,\ldots,M$,
are defined as potentials of 1-forms as
\begin{eqnarray}
  d\calS_\alpha(p) = \theta + \zeta_\alpha(p)dz_\alpha(p),
\end{eqnarray}
where
\begin{eqnarray*}
  \theta = \sum_{n=1}^\infty \Omega_{0n}(p)dt_{0n}
    + \sum_{a=1}^M \sum_{n=0}^\infty
      \Omega_{a n}(p)dt_{a n}.
\end{eqnarray*}
They have Laurent expansions of the form
\begin{eqnarray}
\begin{aligned}
  \calS_0(p) &= \sum_{n=1}^\infty t_{0n}z_0(p)^n
    + t_{00}\log z_0(p)
    - \sum_{n=1}^\infty \frac{z_0(p)^{-n}}{n}v_{0n}, \\
  \calS_a(p) &= \sum_{n=1}^\infty t_{an}z_a(p)^n
    + t_{a0}\log z_a(p) + \phi_a
    - \sum_{n=1}^\infty \frac{z_a(p)^{-n}}{n}v_{an}.
\end{aligned}
\end{eqnarray}

\paragraph*{Implications of $S$-functions}
Let us define $S_\alpha(z)$, $\alpha = 0,1,\ldots,M$, as
\begin{eqnarray}
\begin{aligned}
  S_0(z) &= \sum_{n=1}^\infty t_{0n}z^n + t_{00}\log z
    - \sum_{n=1}^\infty \frac{z^{-n}}{n}v_{0n}, \\
  S_a(z) &= \sum_{n=1}^\infty t_{an}z^n
    + t_{a 0}\log z + \phi_a
    - \sum_{n=1}^\infty \frac{z^{-n}}{n}v_{an}.
\end{aligned}
\end{eqnarray}
$\calS_\alpha(p)$ can be thereby expressed as
\begin{eqnarray*}
  \calS_0(p) = S_0(z_0(p)), \quad
  \calS_a(p) = S_a(z_a(p)).
\end{eqnarray*}
Moreover, the defining equations of $\calS_\alpha(p)$
imply the equations
\begin{eqnarray*}
  \zeta_\alpha(p) = S_\alpha'(z_\alpha(p)),
\end{eqnarray*}
where the prime denotes the derivative with respect to $z$,
and
\begin{eqnarray*}
  \Omega_{\alpha n}(p)
  = \left.\rd_{\alpha n}S_\beta(z)\right|_{z=z_\beta(p)},\quad
  \beta = 0,1,\ldots,M.
\end{eqnarray*}
The former is just a restatement of the Laurent expansion
of $\zeta_\alpha(p)$.  The latter implies that
$\Omega_{\alpha n}(p)$ can be written in several
different forms as
\begin{eqnarray}
  \Omega_{0n}(p)
= \begin{cases}
  \displaystyle
  z_0(p)^n - \sum_{m=1}^\infty \frac{z_0(p)^{-m}}{m}\rd_{0n}v_{0m},\\
  \displaystyle
  \rd_{0n}\phi_b
  - \sum_{m=1}^\infty \frac{z_b(p)^{-m}}{m}\rd_{0n}v_{bm},\quad
  b = 1,\ldots,M
  \end{cases}
\\
  \Omega_{an}(p)
= \begin{cases}
  \displaystyle
  - \sum_{m=1}^\infty \frac{z_0(p)^{-m}}{m}\rd_{an}v_{0m},\\
  \displaystyle
  \delta_{ab}z_b(p)^n + \rd_{an}\phi_b
  - \sum_{m=1}^\infty \frac{z_b(p)^{-m}}{m}\rd_{an}v_{bm},\quad
  b = 1,\ldots,M
  \end{cases}
\end{eqnarray}
for $n = 1,2,\ldots$, and
\begin{eqnarray}
  \Omega_{a0}(p)
= \begin{cases}
  \displaystyle
  - \log z_0(p)
  - \sum_{m=1}^\infty \frac{z_0(p)^{-m}}{m}\rd_{a0}v_{0m},\\
  \displaystyle \delta_{ab}\log z_b(p)
  + \rd_{a0}\phi_b
  - \sum_{m=1}^\infty \frac{z_b(p)^{-m}}{m}\rd_{a0}v_{bm},\quad
  b = 1,\ldots,M.
  \end{cases}
\end{eqnarray}
In particular, since $\Omega_{01}(p) = p$,
we have the identities
\begin{eqnarray}
\begin{aligned}
  p &= z_0(p)
      - \sum_{m=1}^\infty
        \frac{z_0(p)^{-m}}{m}\rd_{01}v_{0m}, \\
  p &= \rd_{01}\phi_b
      - \sum_{m=1}^\infty
        \frac{z_b(p)^{-m}}{m}\rd_{01}v_{b m},\quad
      b = 1,\ldots,M,
\end{aligned}
\end{eqnarray}
which imply that the inverse functions
$p = p_0(z)$ and $p = p_b(z)$ of
$z = z_0(p)$ and $z = z_b(p)$
are given explicitly by
\begin{eqnarray}
\begin{aligned}
  p_0(z)
  = z - \sum_{m=1}^\infty \frac{z^{-m}}{m}\rd_{01}v_{0m}
  = \rd_{01}S_0(z), \\
  p_b(z)
  = \rd_{01}\phi_b
    - \sum_{m=1}^\infty \frac{z^{-m}}{m}\rd_{01}v_{b m}
  = \rd_{01}S_b(z).
\end{aligned}
\end{eqnarray}
Consequently,
\begin{eqnarray}
  q_a = \rd_{01}\phi_a,\quad
  r_a = - \rd_{01}v_{a1}.
\end{eqnarray}
Substituting $p = p_\beta(z)$ in
\begin{eqnarray*}
  \left.\rd_{\alpha n}S_\beta(z)\right|_{z=z_\beta(p)}
  = \Omega_{\alpha n}(p)
\end{eqnarray*}
leads to the Hamilton-Jacobi equations
\begin{eqnarray}
  \rd_{\alpha n}S_\beta(z)
  = \Omega_{\alpha n}(\rd_{01}S_\beta(z)).
\end{eqnarray}

\paragraph*{$F$-function}

The $F$-function is defined by the equation
\begin{eqnarray}\label{eq2_22_5}
\begin{aligned}
&\rd_{0n}F = v_{0n}, \quad
  \rd_{an}F = v_{an}, \quad n = 1,2,\ldots,\\
&\rd_{a0}F = - \phi_a
    + \sum_{b=1}^{a-1} t_{b0}\log(-1),
  \quad a = 1,\ldots,M.
\end{aligned}
\end{eqnarray}
The last part containing $\log(-1)$ is slightly different
from the definition of the Madrid group
\cite{MAMM05,MMMA05}  and the previous paper \cite{TT06} of the first two authors, but this is due to arbitrariness
of the $F$-function.  With the $F$-function,
the $S$-functions can be written as
\begin{eqnarray}
\begin{aligned}
  S_0(z)
  &= \sum_{n=1}^\infty t_{0n}z^n + t_{00}\log z
   -  D_0(z)F, \\
  S_a(z)
  &= \sum_{n=1}^\infty t_{an}z^n
     + t_{a0}\log z + \phi_a - D_a(z)F,
\end{aligned}
\end{eqnarray}
where $D_0(z)$ and $D_a(z)$ denote
the following differential operators:
\begin{eqnarray*}
  D_0(z)
  = \sum_{n=1}^\infty \frac{z^{-n}}{n}\rd_{0n}, \quad
  D_a(z)
  = \sum_{n=1}^\infty \frac{z^{-n}}{n}\rd_{an}.
\end{eqnarray*}

\paragraph*{Generalized Faber polynomials and Grunsky coefficients}
The Hamiltonians $\Omega_{\alpha n}(p)$ of the Lax equations
can also be characterized by the generating functions
\begin{eqnarray}
\begin{aligned}
  \log\frac{p_0(z) - q}{z}
  &= - \sum_{n=1}^\infty \frac{z^{-n}}{n}\Omega_{0n}(q), \\
  \log\frac{q - p_a(z)}{q - q_a}
  &= - \sum_{n=1}^\infty \frac{z^{-n}}{n}\Omega_{an}(q).
\end{aligned}
\label{Omega=Faber}
\end{eqnarray}
The left hand sides of these identities are understood
to be rewritten
\begin{eqnarray*}
  \log\frac{p_0(z) - q}{z}
  = \log\frac{p_0(z)}{z}
    + \log\Bigl(1 - \frac{q}{p_0(z)}\Bigr)
\end{eqnarray*}
and
\begin{eqnarray*}
  \log\frac{q - p_a(z)}{q - q_a}
  = \log\Bigl(1 - \frac{p_a(z) - q_a}{q - q_a}\Bigr)
\end{eqnarray*}
and expanded to power series of $q$ and $(q-q_a)^{-1}$,
respectively.

The generalized Grunsky coefficients $b_{ambn} = b_{bnam}$
are defined by the generating functions
\begin{eqnarray}
\begin{aligned}
\log\frac{p_0(z)-p_0(w)}{z-w}
  &= - \sum_{m,n=1}^\infty z^{-m}w^{-n}b_{0m0n},\\
\log\frac{p_0(z)-p_a(w)}{z}
  &= - \sum_{m=1}^\infty\sum_{n=0}^\infty
       z^{-m}w^{-n}b_{0ma n},\\
\log\frac{zw(p_a(z)-p_a(w))}{w-z}
  &= - \sum_{m,n=0}^\infty z^{-m}w^{-n}b_{aman},\\
\log\frac{p_a(z)-p_b(w)}{\epsilon_{ab}}
  &= - \sum_{m,n=0}^\infty z^{-m}w^{-n}b_{ambn}
  \quad (a \not= b).
\end{aligned}
\label{Grunsky}
\end{eqnarray}
They are related to the $F$-function as
\begin{eqnarray}
  \hat{\rd}_{\alpha m}\hat{\rd}_{\beta n} F
  = - b_{\alpha m\beta n}
  \quad (\alpha,\beta = 0,1,\ldots,N),
\end{eqnarray}
where
\begin{eqnarray*}
  \epsilon_{ab}
  = \begin{cases}
    +1 & (a \le b)\\
    -1 & (a > b)
    \end{cases}, \qquad
  \hat{\rd}_{\alpha n}
  = \begin{cases}
    \frac{1}{n}\rd_{\alpha n} & (n \not= 0), \\
    \rd_{\alpha 0} &(n = 0).
    \end{cases}
\end{eqnarray*}

\section{Riemann-Hilbert problem and non-degenerate solutions}

Following the work of the Madrid group \cite{MMMA05,MAMM05},
we now formulate a Riemann-Hilbert problem.
Choose a set of positively oriented simple closed curves $C_1,\ldots,C_M$
and let $D_1,\ldots,D_M$ denote their inside domains.
The Riemann-Hilbert data consist of $M$ pairs $(f_a,g_a)$,
$a = 1,\ldots,M$, of holomorphic functions $f_a = f_a(p,t_{01})$,
$g_a = g_a(p,t_{01})$ of $p,t_{01}$ (defined in a suitable domain)
that satisfy the conditions
\begin{eqnarray}
  \{f_a,g_a\}=  \frac{\rd f_a}{\rd p}\frac{\rd g_a}{\rd t_{01}}
  - \frac{\rd f_a}{\rd t_{01}}\frac{\rd g_a}{\rd p}
  = 1,
\end{eqnarray}
thus defining two-dimensional canonical transformations.
The problem is to seek $M+1$ pairs
$(z_\alpha(p),\zeta_\alpha(p))$, $\alpha = 0,1,\ldots,M$,
of functions of $p$ and $\boldsymbol{t}
= \{t_{0n} : n = 1,2,\ldots\} \cup
\{t_{an} \,:\,  a =1,\ldots,M,\, n = 0,1,2,\ldots\}$
that satisfy the following conditions:
\begin{itemize}
\item[(i)] $z_0(p)$ and $\zeta_0(p)$ are holomorphic
functions on $\CC\setminus(D_1 \cup\ldots\cup D_M)$,
$z_0(p)$ is univalent therein (in particular,
$z'_0(p)$ does not vanish) and, as $p \to \infty$,
\begin{eqnarray}
\begin{aligned}
  z_0(p) &= p + O(p^{-1}),\\
  \zeta_0(p) &= \sum_{n=1}^\infty nt_{0n}z_0(p)^{n-1}
    + \frac{t_{00}}{z_0(p)} + O(p^{-2}).
\end{aligned}
\end{eqnarray}
\item[(ii)] $z_a(p)$ and $\zeta_a(p)$ are
holomorphic functions on $D_a$ punctured at a point $q_a \in D_a$,
	   $z^{-1}_a(p)$ is univalent on $D_a$ and, as $p \to q_a$,
\begin{eqnarray}
\begin{aligned}
  z_a(p) &= \frac{r_a}{p-q_a} + O(1),\\
  \zeta_a(p) &= \sum_{n=1}^\infty nt_{an}z_a(p)^{n-1}
    + \frac{t_{a0}}{z_a(p)} + O((p-q_a)^2).
\end{aligned}
\end{eqnarray}
$q_a$ and $r_a$ are functions of the time variables
to be thus determined.
\item[(iii)] For $a = 1,\ldots,M$, the four functions
$z_0(p),\zeta_0(p),z_a(p),\zeta_a(p)$ can be
analytically continued to a neighborhood of $C_a$
and satisfy the functional equations
\begin{eqnarray}
  z_a(p) = f_a(z_0(p),\zeta_0(p)), \quad
  \zeta_a(p) = g_a(z_0(p),\zeta_0(p))
\label{RH-fg}
\end{eqnarray}
therein.
\end{itemize}

Functions $z_\alpha(p)$ satisfying above conditions are solutions of the
universal Whitham hierarchy and $\zeta_\alpha(p)$'s are corresponding
Orlov-Schulman functions, as is proved in \cite{MAMM05}, Theorem 1.

Note that formally we can prove the converse. Namely there exist
Riemann-Hilbert data for each solution of the universal Whitham
hierarhcy.
\begin{proposition}
Let $(z_{\alpha}(p)\,:\,\alpha=0,1,\ldots,M)$ be a solution of the
universal Whitham hierarchy, and
$(\zeta_{\alpha}(p)\,:\,\alpha=0,1,\ldots,M)$ the corresponding
Orlov-Schulman functions. For $a=1,\ldots, M$, there exist functions
$f_{a}(p,t_{01})$ and $g_a(p,t_{01})$ such that
\begin{equation}
\label{eq1_21_9}
\begin{split}
 z_{a}=&f_{ a}(z_{0},\zeta_{0}),\qquad
 \zeta_{a}=g_{ a}(z_{0},\zeta_{0}),
\end{split}
\end{equation}
and
$$
  \left\{f_{a}(p,t_{01}), g_{a}(p,t_{01})\right\}=1.
$$
\end{proposition}

\begin{proof}
This is the same as Propositions 4 and 5 of
\cite{MMMA05}, but let us prove it here in our language as in
\cite{TT95}.
Given a solution
$(z_{\alpha}(p)\,:\,\alpha=0,1,\ldots,M)$ of the universal Whitham
hierarchy, construct the dressing functions $\varphi_{\alpha}(p)$ as
given by Proposition \ref{p1}. 
(Recall that we have put $q^{(0)}_a=0$.)
For any $\alpha$, let
\begin{equation*}
\begin{split}
 \tilde{f}_{\alpha}(p,t_{01})&=
  \exp\left(
    -\text{ad}\,\varphi_{\alpha}(\tilde{\boldsymbol{t}}=0)
  \right) p,
\\
 \tilde{g}_{\alpha}(p,t_{01})&=
  \exp\left(
    -\text{ad}\,\varphi_{\alpha}(\tilde{\boldsymbol{t}}=0)
  \right)t_{01}
\end{split}
\end{equation*}
where $\tilde{\boldsymbol{t}}=\boldsymbol{t}\setminus\{t_{01}\}$.
Notice that
\begin{equation*}
\begin{split}
 z_{0}(p, \tilde{\boldsymbol{t}}=0)
 &=\exp\left(
   \text{ad}\,\varphi_0(\tilde{\boldsymbol{t}}=0)
  \right)p,
\\
 \zeta_{0}(p, \tilde{\boldsymbol{t}}=0)
 &=\exp\left(
    \text{ad}\,\varphi_0(\tilde{\boldsymbol{t}}=0)
  \right)t_{01},
\\
 z_{a}^{-1}(p, \tilde{\boldsymbol{t}}=0)
 &=\exp\left(
    \text{ad}\,\varphi_{a}(\tilde{\boldsymbol{t}}=0)
  \right)p,
\\
 (-z_{a}^2\zeta_{a})(p, \tilde{\boldsymbol{t}}=0)
 &=\exp\left(
    \text{ad}\,\varphi_a(\tilde{\boldsymbol{t}}=0)
  \right)t_{01}.
\end{split}
\end{equation*}
Therefore,
\begin{equation*}
\begin{split}
 &\tilde{f}_{a}
  \left(
    z_{a}^{-1}(p, \tilde{\boldsymbol{t}}=0),
    (-z_{a}^2\zeta_{a})(p, \tilde{\boldsymbol{t}}=0)
  \right)
\\
 =&
 \tilde{f}_0
  \left(
    z_{0}(p, \tilde{\boldsymbol{t}}=0),
    \zeta_{0}(p, \tilde{\boldsymbol{t}}=0)
  \right)
  =p,
\\
 &\tilde{g}_{a}
  \left(
    z_{a}^{-1}(p, \tilde{\boldsymbol{t}}=0),
    (-z_{a}^2\zeta_{a})(p, \tilde{\boldsymbol{t}}=0)
  \right)
\\
 =&
 \tilde{g}_0
 \left(
    z_{0}(p, \tilde{\boldsymbol{t}}=0),
    \zeta_{0}(p, \tilde{\boldsymbol{t}}=0)
 \right)=t_{01}
\end{split}
\end{equation*}
for any $a$. Now
\begin{equation*}
\begin{split}
 \frac{\pa }{\pa t_{\beta n}}
   \tilde{f}_0\left(z_{0},\zeta_{0}\right)
 &=\left\{
    \Omega_{\beta n},
    \tilde{f}_0\left(z_{0},\zeta_{0}\right)
  \right\},
\\
 \frac{\pa }{\pa t_{\beta n}}
   \tilde{f}_{a}\left(z_{a}^{-1},-z_{a}^2\zeta_{a}\right)
 &=\left\{
     \Omega_{\beta n},
     \tilde{f}_{a}\left(z_{a}^{-1},-z_{a}^2\zeta_{a}\right)
\right\},
\end{split}
\end{equation*}
 and similarly for $\tilde{g}_0\left(z_{0},\zeta_{0}\right)$ and
 $\tilde{g}_{a}\left(z_{a}^{-1},-z_{a}^2\zeta_{a}\right)$. Therefore,
\begin{equation*}
\begin{split}
  \left.\frac{\pa }{\pa t_{\beta n}}
  \tilde{f}_0(z_{0},\zeta_{0})
  \right|_{\tilde{\boldsymbol{t}}=0}
  =
  \left.\frac{\pa }{\pa t_{\beta n}}
  \tilde{f}_{a}\left(z_{a}^{-1},-z_{a}^2\zeta_{a}\right)
  \right|_{\tilde{\boldsymbol{t}}=0},
\\
  \left.\frac{\pa }{\pa t_{\beta n}}
  \tilde{g}_0\left(z_{0},\zeta_{0}\right)
  \right|_{\tilde{\boldsymbol{t}}=0}
  =
  \left.\frac{\pa }{\pa t_{\beta n}}
  \tilde{g}_{a}\left(z_{a}^{-1},-z_{a}^2\zeta_{a}\right)
  \right|_{\tilde{\boldsymbol{t}}=0}.
\end{split}
\end{equation*}
In the same way, one can show that
\begin{equation*}
\begin{split}
  \left.
  \frac{\pa }{\pa t_{\beta_k n_k}}\ldots\frac{\pa }{\pa t_{\beta_1 n_1}}
  \tilde{f}_0(z_{0},\zeta_{0})
  \right|_{\tilde{\boldsymbol{t}}=0}
  =
  \left.
  \frac{\pa }{\pa t_{\beta_k n_k}}\ldots\frac{\pa }{\pa t_{\beta_1 n_1}}
  \tilde{f}_{a}\left(z_{a}^{-1},-z_{a}^2\zeta_{a}\right)
  \right|_{\tilde{\boldsymbol{t}}=0},
\\
  \left.
  \frac{\pa }{\pa t_{\beta_k n_k}}\ldots\frac{\pa }{\pa t_{\beta_1 n_1}}
  \tilde{g}_0(z_{0},\zeta_{0})
  \right|_{\tilde{\boldsymbol{t}}=0}
  =
  \left.
  \frac{\pa }{\pa t_{\beta_k n_k}}\ldots\frac{\pa }{\pa t_{\beta_1 n_1}}
  \tilde{g}_{a}\left(z_{a}^{-1},-z_{a}^2\zeta_{a}\right)
  \right|_{\tilde{\boldsymbol{t}}=0}.
\end{split}
\end{equation*}
These show that
$$
  \tilde{f}_0\left(z_{0},\zeta_{0}\right)
  =
  \tilde{f}_{a}\left(z_{a}^{-1},-z_{a}^2\zeta_{a}\right)
\qquad
  \tilde{g}_0\left(z_{0},\zeta_{0}\right)
  =
  \tilde{g}_{a}\left(z_{a}^{-1},-z_{a}^2\zeta_{a}\right).
$$
Notice that by definition,
$$
  \left\{\tilde{f}_{0}(p,t_{01}), \tilde{g}_{0}(p,t_{01})\right\}=1,
\qquad
  \left\{\tilde{f}_{a}(p,t_{01}), \tilde{g}_{a}(p,t_{01})\right\}=1.
$$
 One can solve the equations
\begin{equation*}
  \tilde{f}_0(p,t_{01})
  =
  \tilde{f}_{a}\left(\tilde{p}^{-1},-\tilde{p}^2\tilde{t}_{01}\right),
\qquad
  \tilde{g}_{0}(p,t_{01})
  =
  \tilde{g}_{a}\left(\tilde{p}^{-1}, -\tilde{p}^2\tilde{t}_{01}\right),
\end{equation*}
for $\tilde{p}$ and $\tilde{t}_{01}$, which gives
$$
  \tilde{p}=f_{a}(p,t_{01}),
\qquad
  \tilde{t}_{01}=g_{a}(p,t_{01}).
$$
 This implies \eqref{eq1_21_9}. It is also   straightforward to show
 that $\{f_{a},g_{a}\}=1$.
\end{proof}

We now specialize the Riemann-Hilbert problem to the case where the
canonical transformations are defined by generating functions
$H_a(z_0,z_a)$, $a = 1,\ldots,M$.  The generating functions are assumed
to satisfy the non-degeneracy conditions
\begin{eqnarray}\label{eq3_1_1}
  H_{a,z_0z_a}(z_0,z_a) \not= 0.
\end{eqnarray}
Accordingly, the functional equations (\ref{RH-fg})
connecting the four functions $z_0(p),\zeta_0(p),z_a(p),\zeta_a(p)$
are converted to the generalized string equations
\begin{eqnarray}
  \zeta_0(p) = H_{a,z_0}(z_0(p),z_a(p)), \quad
  \zeta_a(p) = - H_{a,z_a}(z_0(p),z_a(p)).
\label{RH-H}
\end{eqnarray}
Existence of such generating functions $H_a$ under the assumption of
non-degeneracy of $f_a$'s:
\begin{equation}
    \frac{\pa f_{a}}{\pa t_{01}}\neq 0,
\label{fa-non-deg}
\end{equation}
can be proved as in \cite{Teo0906}, \S3.4. In fact, under the assumption
\eqref{fa-non-deg}, we can solve the equation $z_{a}=f_a(z_0,\zeta_0)$
to obtain
\begin{equation*}
\zeta_{0}= A_{a}(z_{0}, z_{a}).
\end{equation*}
Equivalently,
\begin{equation}
\label{eq1_21_6}
    z_a=f_a(z_0,A_a(z_0,z_a)).
\end{equation}
Define $B_{a}(z_{0}, z_{a})$ so that
\begin{equation}
\label{eq1_21_7}
    \zeta_{a}=B_{a}(z_{0}, z_{a})=g_{a}(z_{0},A_{a}(z_{0}, z_{a})).
\end{equation}
Differentiating \eqref{eq1_21_6} with respect to $z_0$ and $z_a$ and
\eqref{eq1_21_7} with respect to $z_0$, we find that
\begin{equation*}
\begin{split}
   0=&
   \frac{\pa f_{a}}{\pa z_{0}}
   +
   \frac{\pa f_{a}}{\pa \zeta_{0}}\frac{\pa A_{a}}{\pa z_{0}}
\hspace{0.5cm}
\Longrightarrow
\hspace{0.5cm}
   \frac{\pa A_{a}}{\pa z_{0}}
   =
   \frac{-\frac{\pa f_{a}}{\pa z_{0}}}
        { \frac{\pa f_{a}}{\pa \zeta_{0}}}
\\
   1= &
   \frac{\pa f_{a}}{\pa \zeta_{0}}\frac{\pa A_{a}}{\pa z_{a}}
\hspace{1.5cm}
   \Longrightarrow
\hspace{0.5cm}
   \frac{\pa A_{a}}{\pa z_{a}}
   =
   \frac{1}{\frac{\pa f_{a}}{\pa \zeta_{0}}},
\\
   \frac{\pa B_{a}}{\pa z_{0}}=&
   \frac{\pa g_{a}}{\pa z_{0}}
   +
   \frac{\pa g_{a}}{\pa \zeta_{0}}\frac{\pa A_{a}}{\pa z_{0}}
\hspace{0.5cm}
   \Longrightarrow
\hspace{0.5cm}
   \frac{\pa B_{a}}{\pa z_{0}}=
   \frac{\pa g_{a}}{\pa z_{0}}
   -
   \frac{\pa g_{a}}{\pa \zeta_{0}}
   \frac{ \frac{\pa f_{a}}{\pa z_{0}}}
        {\frac{\pa f_{a}}{\pa \zeta_{0}}}
   =
   -\frac{1}{\frac{\pa f_{a}}{\pa \zeta_{0}}}.
\end{split}
\end{equation*}
Hence we have $\pa_{z_{a}} A_{a} = - \pa_{z_{0}} B_{a}$, which implies
that there exists $H_{a}(z_0,z_a)$ satisfying \eqref{RH-H} and \eqref{eq3_1_1}.

Our goal in the following is to solve the generalized string equations
(\ref{RH-H}) in the language of geometry of the space
\begin{eqnarray*}
  \calZ := \{(z_\alpha(p) \,:\, \alpha=0,\ldots,M) \mid
    \mbox{properties of $z_\alpha(p)$'s in (i), (ii)} \}
\end{eqnarray*}
of the $M+1$-tuple of functions $z_\alpha(p)$, $\alpha = 0,1,\ldots,M$.
This enables us to understand the universal Whitham hierarchy as a
system of integrable commuting flows on $\calZ$, just as achieved in the
case of the dispersionless Toda hierarchy \cite{Teo0906}.

To this end, we define the functions $t_{0n},t_{00},v_{0n}$ ($n
=1,2,\ldots$) and $t_{an},t_{a0},v_{an}$ ($a = 1,\ldots,M$, $n
=1,2,\ldots$) on $\calZ$ as
\begin{eqnarray}
\begin{aligned}
  nt_{0n} &= \sum_{a=1}^M\frac{1}{2\pi i}\oint_{C_a}
             H_{a,z_0}(z_0(p),z_a(p))z_0(p)^{-n}dz_0(p),\\
  t_{00} &= \sum_{a=1}^M\frac{1}{2\pi i}\oint_{C_a}
             H_{a,z_0}(z_0(p),z_a(p))dz_0(p),\\
  v_{0n} &= \sum_{a=1}^M\frac{1}{2\pi i}\oint_{C_a}
            H_{a,z_0}(z_0(p),z_a(p))z_0(p)^ndz_0(p)
\end{aligned}
\label{t0v0-oint}
\end{eqnarray}
and
\begin{eqnarray}
\begin{aligned}
  nt_{an} &= \frac{1}{2\pi i}\oint_{C_a}
               H_{a,z_a}(z_0(p),z_a(p))z_a(p)^{-n}dz_a(p),\\
  t_{a0} &=  \frac{1}{2\pi i}\oint_{C_a}
              H_{a,z_a}(z_0(p),z_a(p))dz_a(p),\\
  v_{an} &=  \frac{1}{2\pi i}\oint_{C_a}
              H_{a,z_a}(z_0(p),z_a(p))z_a(p)^ndz_a(p).
\end{aligned}
\label{tava-oint}
\end{eqnarray}
This is just a restatement of the string equations (\ref{RH-H}).
$t_{00}$ and $t_{a0}$'s are automatically constrained as
\begin{eqnarray*}
  t_{00} = - \sum_{a=1}^\infty t_{a0}.
\end{eqnarray*}
The contour integrals on the right hand side
of (\ref{t0v0-oint}) are derived by continuously
deforming a simple closed curve $C_\infty$ encircling $p = \infty$ and
separating it from all $D_a$'s. Notice that since $z_a(q_a)=\infty$,
$z_a(p)$ maps the inside of $D_a$ onto the outside of $z_a(C_a)$.
Therefore,
\begin{eqnarray*}
  \oint_{C_a} z_a(p)^mz_a'(p)dp=-\delta_{m,-1}.
\end{eqnarray*}

In the next section, we shall reconstruct $\rd_{\alpha n}$'s as globally
defined vector fields on $\calZ$, and show that $t_{\alpha n}$'s may be
thought of as ``dual'' (local) coordinates on $\calZ$ with respect to
these vector fields.  This is the same geometric situation as observed
in the case of the dispersionless Toda hierarchy \cite{Teo0906}.  The
universal Whitham hierarchy is thus realized as a system of commuting
flows on $\calZ$.  This geometric setting can be cast into the usual
setting in the $\bst$ space by the inverse of the period map
$(z_\alpha(p) : \alpha = 0,1,\ldots,M) \mapsto \bst$.  The functions
$z_\alpha(p)$ and $\zeta_\alpha(p)$ on $\calZ$ are pulled back by this
inverse period map to become a solution of the string equations
(\ref{RH-H}), hence a solution of the universal Whitham hierarchy.

The $S$-function (in particular $\phi_a$) and the $F$-function, too, can
be primarily defined as a function on $\calZ$, then pulled back to the
$\bst$ space.  We shall discussed this issue in later sections.

\section{Construction of vector fields $\rd_{\alpha n}$ on $\calZ$}

Following \cite{Teo0906}, we reconstruct $\rd_{\alpha n}$'s as
vector fields on $\calZ$.

\begin{theorem}
\label{thm:vec-field}
If the vector fields $\rd_{0n}$ ($n = 1,2,\ldots$)
and $\rd_{an}$ ($a = 1,\ldots,M$, $n = 0,1,2,\ldots$)
on $\calZ$ satisfy the equations
\begin{eqnarray}\label{iden}
  \frac{\rd_{\alpha n}z_b(p)}{z_b'(p)}
  - \frac{\rd_{\alpha n}z_0(p)}{z_0'(p)}
  = \frac{\Omega_{\alpha n}'(p)}
         {z_0'(p)z_b'(p)H_{b,z_0z_b}(z_0(p),z_b(p))}
\label{rd(z)}
\end{eqnarray}
on $C_b$ for $b = 1,\ldots,M$, where the primes denote the derivatives
with respect to $p$, then they act on $t_{\beta m}$ ($m=0,1,2,\dots$) as
\begin{eqnarray}
  \rd_{\alpha n}t_{\beta m}
    = \delta_{\alpha\beta}\delta_{nm}
\label{rd(t)}
\end{eqnarray}
and on $v_{\beta m}$ ($m=1,2,\dots$) as
\begin{eqnarray}
  \rd_{\alpha n}v_{\beta m}
    = \begin{cases}
       - nmb_{\alpha n\beta m} & (n \not= 0)\\
       - mb_{\alpha 0\beta m}  & (n = 0)
      \end{cases}
\label{rd(v)}
\end{eqnarray}
\end{theorem}

\begin{remark}
(\ref{rd(t)}) implies that $t_{\alpha n}$'s may be thought of as a
system of local coordinates on $\calZ$.  (\ref{rd(v)}) shows that the
vector fields $\rd_{\alpha n}$ correspond to the time evolutions of the
universal Whitham hierarchy.
\end{remark}
\begin{remark}
$\rd_{\alpha n}z_\beta(p)$'s are uniquely determined by (\ref{rd(z)}).
Though this is an implication of (\ref{rd(t)}) and (\ref{rd(v)}), one
can directly confirm it as follows.  Let $Z_b(p),Z_0(p)$ and $W_b(p)$
denote the three terms in (\ref{rd(z)}).  Consequently, they satisfy the
equations
\begin{eqnarray}
  Z_b(p) - Z_0(p) = W_b(p)
\end{eqnarray}
for $b = 1,\ldots,M$ in a neighborhood of $C_b$.  As holomorphic
functions, $Z_b(p),Z_0(p)$ are extended to $D_b$ and
$\CP^1\setminus(D_1\cup\cdots\cup D_M)$ respectively, and behave as
\begin{eqnarray*}
  Z_b(p) = O(1) \quad (p \to q_b),\quad
  Z_0(p) = O(p^{-1}) \quad (p \to \infty).
\end{eqnarray*}
One can decompose $Z_0(p)$
in $\CP^1\setminus(D_1 \cup\cdots\cup D_M)$ as
\begin{eqnarray*}
  Z_0(p) = \sum_{a=1}^M Z_{0a}(p), \quad
  Z_{0a}(p) = \frac{1}{2\pi i}\oint_{C_a}\frac{Z_0(q)}{q-p}dq.
\end{eqnarray*}
$Z_{0a}(p)$ is a holomorphic function in
$\CP^1\setminus D_a$, $O(p^{-1})$ as $p \to \infty$,
and can be continued to a neighborhood of $C_a$
by deforming the contour $C_a$ inward.
The foregoing equation for $Z_b(p)$ and $Z_0(p)$
can be thereby rewritten as
\begin{eqnarray*}
  \left(Z_b(p) - \sum_{a\neq b}Z_{0a}(p)\right)
  - Z_{0b}(p)
  = W_b(p).
\end{eqnarray*}
One can consider this equation as splitting
$W_b(p)$ into a sum of holomorphic functions
$W_{b+}(p)$ and $W_{b-}(p)$ defined
in $D_b$ and in $\CP^1\setminus D_b$, respectively.
In particular,
\begin{eqnarray}
  Z_{0b}(p)
  = - W_{b-}(p)
  = - \frac{1}{2\pi i}\oint_{C_b}\frac{W_b(q)}{p-q}dq
  \quad (p \in \CP^1\setminus D_b),
\end{eqnarray}
hence
\begin{eqnarray}
  Z_0(p) = - \sum_{a=1}^M \frac{1}{2\pi i}\oint_{C_a}\frac{W_a(q)}{p-q}dq
  \quad (p \in \CP^1\setminus(D_1 \cup\cdots\cup D_M)).
\end{eqnarray}
One can find a similar integral formula for $Z_b(p)$
as well.
\end{remark}

\begin{proof}[Proof of Theorem \ref{thm:vec-field}]
Let us first consider the action of $\rd_{\alpha n}$ on
\begin{eqnarray*}
  t_{0m}
  = \sum_{b=1}^M\frac{1}{2\pi im}
     \oint_{C_b}H_{b,z_0}(z_0(p),z_b(p))z_0(p)^{-m}z_0'(p)dp.
\end{eqnarray*}
Applying $\rd_{\alpha n}$ to the integrand,
we have the identity
\begin{eqnarray*}
\begin{aligned}
&\rd_{\alpha n}\left(H_{b,z_0}(z_0(p),z_b(p))z_0(p)^{-m}z_0'(p) \right) \\
&= \frac{\rd}{\rd p}
   \left( H_{b,z_0}(z_0(p),z_b(p))z_0(p)^{-m}\rd_{\alpha n}z_0(p)\right)\\
&\quad
   + H_{b,z_0z_b}(z_0(p),z_b(p))z_0'(p)z_b'(p)
     \left(\frac{\rd_{\alpha n}z_b(p)}{z_b'(p)}
          - \frac{\rd_{\alpha n}z_0(p)}{z_0'(p)}\right)z_0(p)^{-m}\\
&= \frac{\rd}{\rd p}
   \left( H_{b,z_0}(z_0(p),z_b(p))z_0(p)^{-m}\rd_{\alpha n}z_0(p)\right)
   + \Omega_{\alpha n}'(p)z_0(p)^{-m}.
\end{aligned}
\end{eqnarray*}
Note that we have used the assumed equation \eqref{iden}
in the last line. Consequently,
\begin{eqnarray*}
\begin{aligned}
  \rd_{\alpha n}t_{0m}
  &= \sum_{b=1}^M\frac{1}{2\pi im}\oint_{C_b}
     \Omega_{\alpha n}'(p)z_0(p)^{-m}dp \\
  &= \frac{1}{2\pi im}\oint_{C_\infty}
     \Omega_{\alpha n}'(p)z_0(p)^{-m}dp,
\end{aligned}
\end{eqnarray*}
by deforming $C_b$'s to a simple closed curve $C_\infty$ encircling $p =
\infty$.

On the other hand, one can deduce from (\ref{Omega=Faber}) and the first
and the second equations in (\ref{Grunsky}) (and the symmetry of
$b_{\alpha n\beta m}$) the following Laurent expansion of
$\Omega_{\alpha n}(p)$'s with respect to $z_0(p)$:
\begin{eqnarray}\label{eq2_22_3}
\begin{aligned}
  \Omega_{0n}(p)
  &= z_0(p)^n + \sum_{m=1}^\infty nb_{0n0m}z_0(p)^{-m},\\
  \Omega_{an}(p)
  &= \sum_{m=1}^\infty nb_{an0m} z_0(p)^{-m}
  \quad (n \ge 1),\\
  \Omega_{a0}(p)
  &= - \log z_0(p) + \sum_{m=1}^\infty b_{a00m}z_0(p)^{-m}.
\end{aligned}
\end{eqnarray}
Inserting the derivatives
\begin{eqnarray*}
\begin{aligned}
  \Omega_{0n}'(p)
  &= nz_0(p)^{n-1}z_0'(p)
     - \sum_{m=1}^\infty nmb_{0n0m}z_0(p)^{-m-1}z_0'(p),\\
  \Omega_{an}'(p)
  &= - \sum_{m=1}^\infty nmb_{an0m}z_0(p)^{-m-1}z_0'(p)
  \quad (n \ge 1),\\
  \Omega_{a0}'(p)
  &= - \frac{z_0'(p)}{z_0(p)}
     - \sum_{m=1}^\infty mb_{a00m}z_0(p)^{-m-1}z_0'(p)
\end{aligned}
\end{eqnarray*}
into the contour integral, we readily obtain
(\ref{rd(t)})  for $\beta = 0$.

In the same way, the action of $\rd_{\alpha n}$ on
\begin{eqnarray*}
  v_{0m}
  = \sum_{b=1}^M \frac{1}{2\pi i}\oint_{C_b}
    H_{b,z_0}(z_0(p),z_b(p))z_b(p)^mz_b'(p)dp
\end{eqnarray*}
can be expressed as
\begin{eqnarray*}
  \rd_{\alpha n}v_{0m}
  = \frac{1}{2\pi i}\oint_{C_\infty}
      \Omega_{\alpha n}'(p)z_0(p)^mdp.
\end{eqnarray*}
This contour integral, too, can be evaluated
by the foregoing Laurent expansions of
$\Omega_{\alpha n}'(p)$.  We can thus derive
(\ref{rd(v)}) for $\beta = 0$.

Let us now consider the action of $\rd_{\alpha n}$ on
\begin{eqnarray*}
\begin{aligned}
  t_{bm}
  &= \frac{1}{2\pi im}\oint_{C_b}H_{b,z_b}(z_0(p),z_b(p))z_b(p)^{-m}z_b'(p)dp,\\
  t_{b0}
  &= \frac{1}{2\pi i}\int_{C_b}H_{b,z_b}(z_0(p),z_b(p))z_b'(p)dp, \\
  v_{bm}
  &= \frac{1}{2\pi i}\int_{C_b}H_{b,z_b}(z_0(p),z_b(p))z_b(p)^mz_b'(p)dp
\end{aligned}
\end{eqnarray*}
As in the previous case, we can deduce that
\begin{eqnarray}\label{eq2_22_1}
\begin{aligned}
  \rd_{\alpha n}t_{bm}
  &= - \frac{1}{2\pi im}\oint_{C_b}\Omega_{\alpha n}'(p)z_b(w)^{-m}dp,\\
  \rd_{\alpha n}t_{b0}
  &= - \frac{1}{2\pi i}\oint_{C_b}\Omega_{\alpha n}'(p)dp,\\
  \rd_{\alpha n}v_{bm}
  &= - \frac{1}{2\pi i}\oint_{C_b}\Omega_{\alpha n}'(p)z_b(p)^mdp,
\end{aligned}
\end{eqnarray}
We can now use the following Laurent expansion
of $\Omega_{\alpha n}(p)$'s with respect to $z_b(p)$ derived from
 \eqref{Omega=Faber} and the second, third and fourth equations of
 \eqref{Grunsky}:
\begin{eqnarray}\label{eq2_22_4}
\begin{aligned}
  \Omega_{0n}(p)
  &= \sum_{m=0}^\infty nb_{0nbm}z_b(p)^{-m},\\
  \Omega_{an}(p)
  &= \delta_{ab}z_b(p)^n + \sum_{m=0}^\infty nb_{anbm}z_b(p)^{-m},
     \quad (n \ge 1),\\
  \Omega_{a0}(p)
  &= \begin{cases}
     - \log\epsilon_{ba} + \sum_{m=0}^\infty b_{a0bm}z_b(p)^{-m}
        & (b \not= a), \\
     \log z_a(p) + \sum_{m=0}^\infty b_{a0am}z_a(p)^{-m}
        & (b = a).
     \end{cases}
\end{aligned}
\end{eqnarray}
Inserting their derivatives
\begin{eqnarray*}
\begin{aligned}
  \Omega_{0n}'(p)
  &= - \sum_{m=0}^\infty nmb_{0nbm}z_b(p)^{-m-1}z_b'(p),\\
  \Omega_{an}'(p)
  &= \delta_{ab}nz_b(p)^{n-1}z_b'(p)
     - \sum_{m=0}^\infty nmb_{anbm}z_b(p)^{-m-1}z_b'(p)
     \quad (n \ge 1),\\
  \Omega_{a0}'(p)
  &= \delta_{ab}\frac{z_b'(p)}{z_b(p)}
     - \sum_{m=0}^\infty mb_{a0bm}z_b(p)^{-m-1}z_b'(p)
\end{aligned}
\end{eqnarray*}
into the contour integrals \eqref{eq2_22_1}, we can confirm
the remaining parts of (\ref{rd(t)}) and (\ref{rd(v)}).
This completes the proof of the theorem.
\end{proof}

\section{Construction of $\phi_a$'s}
We construct the Phi functions $\phi_a$, $a=1,\ldots,M$, as follows:
\begin{multline}
\label{eq2_2_2}
   \phi_a
   =
   \sum_{b=1}^M t_{b0}b_{a0b0}
   +
   \sum_{\gamma=0}^M\sum_{m=1}^{\infty}mt_{\gamma m} b_{a0 \gamma m}
   +
   \sum_{b=1}^{a-1} t_{b0}\log(-1)
\\
   -
   \frac{1}{2\pi i}\sum_{b=1}^M
   \oint_{C_b}\frac{H_b(z_0(p),z_b(p))}{p-q_a}dp.
\end{multline}
\begin{proposition}
The function $\phi_{a}$ defined by \eqref{eq2_2_2}
 satisfies
\begin{equation}\label{eq2_22_2}
\pa_{\beta n} \phi_a  =
   \begin{cases}
    n b_{a 0 \beta n},\qquad &(n\neq 0)
   \\
    b_{a 0 \beta 0}+ \log\epsilon_{ a\beta}, &(n=0)
   \end{cases}.
 \end{equation}
\end{proposition}

\begin{proof}
\begin{multline*}
  \pa_{\beta n}
    \left( \frac{H_b(z_0(p),z_b(p))}{p-q_a} \right)
\\
    =
    \frac{H_{b,z_0}(z_0(p),z_b(p))}{p-q_a}
     \pa_{\beta n}z_0(p)
    +
    \frac{H_{b,z_b}(z_0(p),z_b(p))}{p-q_a}
     \pa_{\beta n}z_b(p)
\\
    +
    \frac{ H_b(z_0(p),z_b(p))}{(p-q_a)^2} \pa_{\beta n} q_a.
\end{multline*}
Therefore,
\[
   \pa_{\beta n}\phi_a
   =
   T_1(\beta,n)+T_2(\beta,n)+
   \begin{cases}
    n b_{a 0\beta n},\qquad &(n\neq 0)
    \\
    b_{a 0 \beta 0}+\log \epsilon_{ a\beta}, &(n=0)
   \end{cases},
\]
where
\begin{equation*}
 \begin{split}
    T_1(\beta,n)
    =&
    \sum_{b=1}^M t_{b0} \pa_{\beta n}   b_{a0b0}
    +
    \sum_{\gamma=0}^M \sum_{m=1}^{\infty}
    mt_{\gamma m} \pa_{\beta n}   b_{ \gamma ma0}
\\
    =&
    \sum_{m=1}^{\infty}mt_{0 m} \pa_{\beta n}   b_{0 ma0}
    +
    \sum_{b=1}^M\left(
      t_{b0} \pa_{\beta n}   b_{a0b0}
      +
      \sum_{m=1}^{\infty} mt_{b m} \pa_{\beta n}   b_{b ma0}
    \right),
 \end{split}
\end{equation*}
and
\begin{multline*}
   T_2(\beta,n)
   =
   -\frac{1}{2\pi i}\sum_{b=1}^M
   \oint_{C_b}
    \left(
     \frac{H_{b,z_0}(z_0(p),z_b(p))}{p-q_a}
     \pa_{\beta n}  z_0(p)
     \right.
\\
     \left.
     +
     \frac{H_{b,z_b}(z_0(p),z_b(p))}{p-q_a}
      \pa_{\beta n}   z_b(p)
     +
     \frac{ H_b(z_0(p),z_b(p))}{(p-q_a)^2}
      \pa_{\beta n}   q_a
    \right)dp.
\end{multline*}
The goal is to show that $T_1(\beta,n)+T_2(\beta,n)=0$ for all
$(\beta,n)$. By the definition of $t_{0m}$, we have
\begin{multline*}
    \sum_{m=1}^{\infty}
    mt_{0 m} \pa_{\beta n}  b_{0 ma0}
\\
    =\frac{1}{2\pi i}
    \sum_{b=1}^M\oint_{C_b} H_{b,z_0}(z_0(p),z_b(p))z_0'(p)
     \left(
      \sum_{m=1}^{\infty}
       \pa_{\beta n} b_{0 ma0} z_0(p)^{-m}
    \right)dp.
\end{multline*}
Differentiating
\[
   -\log (p -q_a)
   =\Omega_{a 0}(p)
   =-\log z_0(p)+\sum_{m=1}^{\infty}b_{0ma 0}z_0(p)^{-m}
\]
with respect to $t_{\beta n}$, we have
\begin{equation*}
 \begin{split}
    \frac{1}{p-q_a} \pa_{\beta n}  q_a
    =&
    \sum_{m=1}^{\infty}
     \pa_{\beta n} b_{0 ma0}  z_0(p)^{-m}
\\
    &-
    \left(
      \frac{1}{z_0(p)}
      +
      \sum_{m=1}^{\infty}mb_{0ma 0}z_0(p)^{-m-1}
    \right)
     \pa_{\beta n} z_0(p)
\\
     =&
     \sum_{m=1}^{\infty}
       \pa_{\beta n}  b_{0 ma0} z_0(p)^{-m}
     -
      \frac{1}{z_0'(p)}
      \frac{1}{p-q_a} \pa_{\beta n}   z_0(p).
 \end{split}
\end{equation*}
Therefore,
\[
    \sum_{m=1}^{\infty}
     \pa_{\beta n}  b_{0 ma0} z_0(p)^{-m}
    =
    \frac{1}{p-q_a}  \pa_{\beta n} q_a
    +
    \frac{1}{z_0'(p)} \frac{1}{p-q_a} \pa_{\beta n}  z_0(p),
\]
and
\begin{equation*}
 \begin{split}
   \sum_{m=1}^{\infty}
   mt_{0 m} \pa_{\beta n} b_{0 m a0}
   =&
   \frac{1}{2\pi i}\sum_{b=1}^M\oint_{C_b}
    H_{b,z_0}(z_0(p),z_b(p))z_0'(p) \times
\\
   &\qquad \times
   \left(
    \frac{1}{p-q_a}  \pa_{\beta n} q_a
    +
    \frac{1}{z_0'(p)} \frac{1}{p-q_a}  \pa_{\beta n}  z_0(p)
   \right)dp
\\
    =&
    \frac{1}{2\pi i}\sum_{b=1}^M\oint_{C_b}
    \left(
      \frac{H_{b,z_0}(z_0(p),z_b(p))z_0'(p)}{p-q_a}
      \pa_{\beta n} q_a
    \right.
\\
    &\hskip2cm\left.
      +
      \frac{H_{b,z_0}(z_0(p),z_b(p))}{p-q_a}
      \pa_{\beta n} z_0(p)
    \right)dp.
 \end{split}
\end{equation*}
In a similar way, the definition of $t_{b m}$ gives
\begin{multline*}
    t_{b0} \pa_{\beta n} b_{a0b0}
    +\sum_{m=1}^{\infty}mt_{b m} \pa_{\beta n} b_{b m a0}
\\
    =
    \frac{1}{2\pi i}\oint_{C_b}
    H_{b,z_b}(z_0(p),z_b(p))z_b'(p)
    \left(
       \sum_{m=0}^{\infty}
       \pa_{\beta n} b_{b m a0} z_b(p)^{-m}
    \right)dp.
\end{multline*}
Differentiating
\[
    -\log(p-q_a)
    = \Omega_{a 0}(p )
    = \delta_{ab}\log z_b(p)
    + \sum_{m=0}^{\infty}b_{b m a 0}z_b(p)^{-m}
    - \log\epsilon_{ba}
\]
with respect to $t_{\beta n}$ and comparing it with $\Omega'_{a0}(p)$,
we find that
\[
    \frac{1}{p-q_a} \pa_{\beta n} q_a
    =
    \sum_{m=0}^{\infty}
    \pa_{\beta n}b_{b m a0} z_b(p)^{-m}
    -
    \frac{1}{z_b'(p)}\frac{1}{p-q_a} \pa_{\beta n} z_b(p).
\]
Therefore,
\begin{equation*}
 \begin{split}
    &t_{b0} \pa_{\beta n} b_{a0b0}
    +
    \sum_{m=1}^{\infty}
    mt_{b m} \pa_{\beta n} b_{b m a0}
\\
    =&
    \frac{1}{2\pi i}\oint_{C_b}
    \left(
     \frac{H_{b,z_b}(z_0(p),z_b(p))z_b'(p)}{p-q_a}
     \pa_{\beta n} q_a
     +
     \frac{H_{b,z_b}(z_0(p),z_b(p))}{p-q_a}
     \pa_{\beta n} z_b(p)
    \right)dp.
 \end{split}
\end{equation*}
Therefore,
\begin{equation*}
 \begin{split}
    &T_1(\beta,n)+T_2(\beta,n)
\\
    =&
    \frac{1}{2\pi i}\sum_{b=1}^M\oint_{C_b}
    \left(
     \frac{H_{b,z_0}(z_0(p),z_b(p))z_0'(p)}{p-q_a}
     \pa_{\beta n} q_a
    \right.
\\
    &\qquad +\left.
     \frac{H_{b,z_b}(z_0(p),z_b(p))z_b'(p)}{p-q_a}
     \pa_{\beta n} q_a
     -
     \frac{ H_b(z_0(p),z_b(p))}{(p-q_a)^2}
     \pa_{\beta n} q_a \right)dp
\\
    =&
    \frac{1}{2\pi i}\sum_{b=1}^M\oint_{C_b}
    \frac{\pa}{\pa p}
    \left(
     \frac{H_b(z_0(p),z_b(p))}{p-q_a}
    \right)dp \times
    \pa_{\beta n} q_a =0.
 \end{split}
\end{equation*}
This completes the proof.
\end{proof}

Define
\begin{equation}\label{eq2_23_1}\begin{split}
   v_{a0}
   =&
   -\phi_{a}
   +
   \sum_{b=1}^{a-1} t_{b0}\log(-1)
\\
   =&
   -\sum_{b=1}^M t_{b0}b_{a0b0}
   -\sum_{\gamma=0}^M\sum_{m=1}^{\infty}mt_{\gamma m} b_{ a0 \gamma m}
\\
   &+\frac{1}{2\pi i}
    \sum_{b=1}^M\oint_{C_b}\frac{H_b(z_0(p),z_b(p))}{p-q_a}dp.
\end{split}\end{equation}
Then \eqref{eq2_22_2} implies that
\begin{equation}
\label{eq2_2_3}
    \pa_{\beta n}  v_{a0}
    =
    \begin{cases}
    -nb_{ a0 \beta n},\qquad&(n\neq 0),\\
    -b_{a0 \beta 0},\qquad&(n=0)
    \end{cases}.
\end{equation}

\section{Construction of the free energy $F$}
Let $J_{a,1}(z_0,z_a)$ and $J_{a,2}(z_0,z_a)$ be defined so that
\begin{equation}
    -\pa_{z_a}J_{a,1}(z_0,z_a)
    =
    \pa_{z_0}J_{a,2}(z_0,z_a)
    =
    H_a(z_0,z_a)H_{a,z_0z_a}(z_0,z_a).
    \label{J-function}
\end{equation}
We construct the $F$ function as follows:
\begin{equation}
\label{eq2_2_1}
 \begin{split}
   F =&
     \frac{1}{2}\sum_{a=1}^M t_{a0}v_{a0}
    +\frac{1}{2}\sum_{\alpha=0}^M
                \sum_{n=1}^{\infty}t_{\alpha n}v_{\alpha n}
\\
    &
    +\frac{1}{8\pi i}\sum_{a=1}^M\oint_{C_a}
    \Bigl\{
    J_{a,1}(z_0(p),z_a(p))z_0'(p)+J_{a,2}(z_0(p),z_a(p))z_a'(p)
    \Bigr\}dp.
 \end{split}
\end{equation}

\bigskip
\noindent
\begin{proposition}
The $F$ function defined by \eqref{eq2_2_1} satisfies
$$
   \pa_{\beta n} F =v_{\beta n}.
$$
\end{proposition}

\begin{proof}
 A direct computation shows that
\begin{equation*}
 \begin{split}
  &
 \pa_{\beta n}
  \Bigl\{
    J_{a,1}(z_0(p),z_a(p))z_0'(p)+J_{a,2}(z_0(p),z_a(p))z_a'(p)
  \Bigr\}
\\
  =&
  \frac{\pa}{\pa p}
   \Bigl\{
     J_{a,1}(z_0(p),z_a(p)) \pa_{\beta n} z_0
     +
     J_{a,2}(z_0(p),z_a(p)) \pa_{\beta n} z_a
   \Bigr\}
\\
   & -2H_a(z_0(p),z_a(p))
     (\pa_{\beta n} z_a(p)\, z_0'(p) - z_a'(p)\, \pa_{\beta n} z_0(p) )
\\
  =&
  \frac{\pa}{\pa p}
   \Bigl\{
     J_{a,1}(z_0(p),z_a(p)) \pa_{\beta n} z_0
     +
     J_{a,2}(z_0(p),z_a(p)) \pa_{\beta n} z_a
   \Bigr\}
\\
   & -2H_a(z_0(p),z_a(p))\Omega_{\beta n}'(p)
 \end{split}
\end{equation*}
by using the definition of the vector field \eqref{rd(z)}.
Hence,
\begin{equation*}
  \pa_{\beta n} F
   =I_1(\beta,n)+I_2(\beta,n).
\end{equation*}
where
\begin{align*}
    I_1(\beta,n) &=
    \frac{v_{\beta n}}{2}
    +
    \frac{1}{2}\sum_{a=1}^M t_{a0} \pa_{\beta n} v_{a0}
    +
    \frac{1}{2}
     \sum_{\alpha=0}^M\sum_{m=1}^{\infty}
       t_{\alpha m} \pa_{\beta n} v_{\alpha m},
\\
    I_2(\beta,n) &=
    -\frac{1}{4\pi i}
     \sum_{a=1}^M\oint_{C_a}H_a(z_0(p),z_a(p))\Omega_{\beta n}'(p)dp.
\end{align*}

Now if $n\neq 0$, integration by parts shows that
\begin{equation*}
 \begin{split}
    I_2(\beta,n)
    =&
    -\frac{1}{4\pi i}\sum_{a=1}^M\oint_{C_a}
     H_a(z_0(p),z_a(p))\Omega_{\beta n}'(p)dp
\\
    =&
    \frac{1}{4\pi i}\sum_{a=1}^M\oint_{C_a}
    H_{a,z_0}(z_0(p),z_a(p))z_0'(p)\Omega_{\beta n}(p) dp
\\
    &+
    \frac{1}{4\pi i}\sum_{a=1}^M\oint_{C_a}
    H_{a,z_a}(z_0(p),z_a(p))z_a'(p)\Omega_{\beta n}(p) dp.
 \end{split}
\end{equation*}
If $\beta=0$, the first equation in \eqref{eq2_22_3} and the first
 equation in \eqref{eq2_22_4} show that
\begin{equation*}
 \begin{split}
    I_2(0,n)
    =&
    \frac{1}{4\pi i}\sum_{a=1}^M\oint_{C_a}
    H_{a,z_0}(z_0(p),z_a(p))z_0'(p) \times
\\
    &\qquad \times\left(
     z_0(p)^n + n\sum_{m=1}^{\infty}  b_{0m 0n}z_0(p)^{-m}
    \right)dp
\\
    &+
    \frac{1}{4\pi i}\sum_{a=1}^M\oint_{C_a}
    H_{a,z_a}(z_0(p),z_a(p))z_a'(p) \times
\\
    &\qquad\times\left(
      n\sum_{m=0}^{\infty}  b_{0n a m }z_a(p)^{-m}
    \right)dp
\\
    =&
    \frac{v_{0n}}{2}
    +
    \frac{1}{2}\sum_{m=1}^{\infty}nmb_{0m0n}t_{0m}
    +
    \frac{1}{2}\sum_{a=1}^Mnb_{0na0}t_{a0}
\\
    &+
    \frac{1}{2}\sum_{a=1}^M\sum_{m=1}^{\infty}nmb_{0nam}t_{am}
\\
    =&
    \frac{v_{0n}}{2}
    -
    \frac{1}{2}\sum_{a=1}^M t_{a0} \pa_{0n} v_{a0}
    -\frac{1}{2}\sum_{\alpha=0}^M\sum_{m=1}^{\infty}
    t_{\alpha m} \pa_{0n} v_{\alpha m},
 \end{split}
\end{equation*}
by the definition of $t_{\alpha n}$ (\ref{t0v0-oint}, \ref{tava-oint}) and
 actions of $\pa_{\alpha n}$ on $v_{\beta m}$ (\ref{rd(v)}, \ref{eq2_2_3}). 
Therefore,
\[
  \pa_{0n}F =I_1(0,n)+I_2(0,n)=v_{0 n}.
\]
If $\beta=b\neq 0$, $n\neq 0$, the second equation in \eqref{eq2_22_3}
 and the second equation in \eqref{eq2_22_4} show that
\begin{equation*}
 \begin{split}
    I_2(b,n)
    =&
    \frac{1}{4\pi i}\sum_{a=1}^M\oint_{C_a}
    H_{a,z_0}(z_0(p),z_a(p))z_0'(p) \times
\\
    &\qquad\times\left(
    n\sum_{m=1}^{\infty}b_{b n0m}z_0(p)^{-m}\right)dp
\\
    &+
    \frac{1}{4\pi i}\sum_{a=1}^M\oint_{C_a}
    H_{a,z_a}(z_0(p),z_a(p))z_a'(p) \times
\\
    &\qquad \times
    \left(
     \delta_{ab}z_a(p)^n
     +
     n\sum_{m=0}^{\infty}b_{b na m}z_a(p)^{-m}
    \right)dp
\\
    =&
    \frac{1}{2}\sum_{m=0}^{\infty}nmb_{b n0m}t_{0m}
    +
    \frac{v_{bn}}{2}+\frac{1}{2}\sum_{a=1}^M nb_{bn a0} t_{a0}
\\
    &+
    \frac{1}{2}\sum_{a=1}^M\sum_{m=1}^{\infty}nm b_{bn am}t_{am}
\\
    =&
    \frac{v_{bn}}{2}
    -
    \frac{1}{2}\sum_{a=1}^M t_{a0} \pa_{bn} v_{a0}
    -
    \frac{1}{2}\sum_{\alpha=0}^M\sum_{m=1}^{\infty}
    t_{\alpha m} \pa_{bn} v_{\alpha m},
 \end{split}
\end{equation*}
again by (\ref{t0v0-oint}, \ref{tava-oint}) and (\ref{rd(v)}, \ref{eq2_2_3}). 
Therefore,
\[
   \pa_{bn}F =I_1(b,n)+I_2(b,n)=v_{b n}.
\]
Now if $\beta=b\neq 0$, $n=0$, we have
\[
    \Omega_{b0}'(p)=-\frac{1}{p-q_{b}}.
\]
Therefore,
\begin{equation*}
 \begin{split}
    I_2(b,0)
    &=
    -\frac{1}{4\pi i}\sum_{a=1}^M\oint_{C_a}
    H_a(z_0(p),z_a(p))\Omega_{b0}'(p)dp
\\
    &=
    \frac{1}{4\pi i}\sum_{a=1}^M\oint_{C_a}
    \frac{H_a(z_0(p),z_a(p))}{p-q_b}dp.
 \end{split}
\end{equation*}
On the other hand,
\begin{equation*}
 \begin{split}
    I_1(b,0)
    =&
   \frac{v_{b0}}{2}
    +
    \frac{1}{2}\sum_{a=1}^M t_{a0} \pa_{b0} v_{a0}
    +
    \frac{1}{2}\sum_{\alpha=0}^M\sum_{m=1}^{\infty}
    t_{\alpha m} \pa_{b0} v_{\alpha m}
\\
    =&
    \frac{v_{b0}}{2}
    -
    \frac{1}{2}\sum_{a=1}^M t_{a0}b_{a0b0}
    -
    \frac{1}{2}\sum_{\alpha=0}^M\sum_{m=1}^{\infty}
      m t_{\alpha m}b_{ \alpha mb0},
 \end{split}
\end{equation*}
by (\ref{rd(v)}, \ref{eq2_2_3}). 
Hence,
\begin{equation*}
 \begin{split}
    &I_1(b,0)+I_2(b,0)
\\
    =&
    \frac{v_{b0}}{2}
    -
    \frac{1}{2}
    \sum_{a=1}^M t_{a0}b_{a0b0}
     -
    \frac{1}{2}
    \sum_{\alpha=0}^M\sum_{m=1}^{\infty} m t_{\alpha m}b_{ \alpha mb0}
\\
    & +
    \frac{1}{4\pi i}\sum_{a=1}^M\oint_{C_a}
    \frac{H_a(z_0(p),z_a(p))}{p-q_b}dp
\\
    =&
    \frac{v_{b0}}{2}+\frac{v_{b0}}{2}=v_{b0}
 \end{split}
\end{equation*}
because of \eqref{eq2_23_1}.
This completes the proof.
\end{proof}
This proposition and the definition of $v_{a0}$ \eqref{eq2_23_1} shows that the $F$ function  indeed satisfies \eqref{eq2_22_5}.

\section{Special String Equations}
In this section, we consider the special case where the generating
functions $H_a(z_0,z_a)$, $a=1,\ldots, M,$ have the form
\begin{equation*}
    H_a(z_0,z_a)=z_0^{\nu_0}z_a^{\nu_a}, \qquad
    \nu_0, \nu_a\in \mathbb{N},
\end{equation*}
so  that the string equations \eqref{RH-H} become
\begin{equation}
\label{eq3_1_2}
 \begin{split}
  \zeta_0(p)=&\nu_0 z_0(p)^{\nu_0-1}z_a(p)^{\nu_a},\\
  \zeta_a(p)=&-\nu_az_0(p)^{\nu_0}z_a(p)^{\nu_a-1}
 \end{split}
\end{equation}
for $p\in C_a$. These string equations were discussed in \cite{MAMM05}.
>From \eqref{eq3_1_2} and \eqref{zeta-expansion}, we have
\begin{equation}
\label{eq3_1_3}
 \begin{split}
  \nu_0 z_0(p)^{\nu_0-1}z_a(p)^{\nu_a}
  =&\sum_{n=1}^\infty nt_{0n}z_0(p)^{n-1}
    + \frac{t_{00}}{z_0(p)}
    + \sum_{n=1}^\infty z_0(p)^{-n-1}v_{0n}, \\
  -\nu_az_0(p)^{\nu_0}z_a(p)^{\nu_a-1}
  =& \sum_{n=1}^\infty nt_{an}z_a(p)^{n-1}
    + \frac{t_{a0}}{z_a(p)}
    + \sum_{n=1}^\infty z_a(p)^{-n-1}v_{an}
 \end{split}
\end{equation}
for $p\in C_a$. The definitions of $t_{\alpha n}$ and $v_{\alpha n}$
\eqref{t0v0-oint} and \eqref{tava-oint} then become
\begin{eqnarray}
 \label{eq3_1_4}
  \begin{aligned}
  nt_{0n} &= \sum_{a=1}^M\frac{\nu_0}{2\pi i}\oint_{C_a}
             z_0(p)^{\nu_0-n-1}z_a(p)^{\nu_a}dz_0(p),\\
  t_{00} &=  \sum_{a=1}^M\frac{\nu_0}{2\pi i}\oint_{C_a}
             z_0(p)^{\nu_0-1}z_a(p)^{\nu_a}dz_0(p),\\
  v_{0n} &= \sum_{a=1}^M\frac{\nu_0}{2\pi i}\oint_{C_a}
             z_0(p)^{\nu_0+n-1}z_a(p)^{\nu_a}dz_0(p)
  \end{aligned}
\end{eqnarray}
and
\begin{eqnarray}
\begin{aligned}
  nt_{an} &= \frac{\nu_a}{2\pi i}\oint_{C_a}
               z_0(p)^{\nu_0}z_a(p)^{\nu_a-n-1}dz_a(p),\\
  t_{a0} &=  \frac{\nu_a}{2\pi i}\oint_{C_a}
              z_0(p)^{\nu_0}z_a(p)^{\nu_a-1} dz_a(p),\\
  v_{an} &=  \frac{\nu_a}{2\pi i}\oint_{C_a}
              z_0(p)^{\nu_0}z_a(p)^{\nu_a+n-1}dz_a(p).
\end{aligned}
\label{tava-oint_2}
\end{eqnarray}
The functions $J_{a,1}(z_0,z_a)$ and $J_{a,2}(z_0,z_a)$
\eqref{J-function} can be chosen to be
\begin{equation*}
  J_{a,1}(z_0,z_a)=-\frac{\nu_0}{2}z_0^{2\nu_0-1}z_a^{2\nu_a},\qquad
  J_{a,2}(z_0,z_a)=\frac{\nu_a}{2}z_0^{2\nu_0}z_a^{2\nu_a-1}.
\end{equation*}
The free energy \eqref{eq2_2_1} then becomes
\begin{equation}
\begin{split}
   F =&
     \frac{1}{2}\sum_{a=1}^M t_{a0}v_{a0}
    +\frac{1}{2}\sum_{\alpha=0}^M
                \sum_{n=1}^{\infty}t_{\alpha n}v_{\alpha n}
\\
    &
    -\frac{\nu_0}{16\pi i}
    \sum_{a=1}^M\oint_{C_a}z_0(p)^{2\nu_0-1}z_a(p)^{2\nu_a}dz_0(p)
    \\
    &
    +\frac{1}{16\pi i}
    \sum_{a=1}^M\oint_{C_a}\nu_az_0(p)^{2\nu_0}z_a(p)^{2\nu_a-1}dz_a(p)
 \end{split}
\end{equation}
Using \eqref{eq3_1_3} and \eqref{eq3_1_4}, we find that
\begin{equation}
 \begin{split}
    &-\frac{\nu_0}{16\pi i}
     \sum_{a=1}^M\oint_{C_a} z_0(p)^{2\nu_0-1}z_a(p)^{2\nu_a}dz_0(p)
\\
    =
    &-\frac{1}{16\pi i}
    \sum_{a=1}^M\oint_{C_a}
     \left(\nu_0 z_0(p)^{\nu_0-1}z_a(p)^{\nu_a}\right)
      z_0(p)^{\nu_0}z_a(p)^{\nu_a}dz_0(p)
\\
    =
    &-\frac{1}{16\pi i}\sum_{a=1}^M\oint_{C_a}
    \left(
     \sum_{n=1}^\infty nt_{0n}z_0(p)^{n-1}
     + \frac{t_{00}}{z_0(p)}
     + \sum_{n=1}^\infty z_0(p)^{-n-1}v_{0n}\right)
    \\
    &\hspace{4cm}\times
      z_0(p)^{\nu_0}z_a(p)^{\nu_a}dz_0(p)
\\
    =
    &-\frac{1}{8\nu_0}
    \left(  2 \sum_{n=1}^\infty nt_{0n}v_{0n} +t_{00}^2\right).
 \end{split}
\end{equation}
Similarly, one can show that
\begin{equation}
 \begin{split}
    \frac{1}{16\pi i}
    \oint_{C_a}\nu_az_0(p)^{2\nu_0}z_a(p)^{2\nu_a-1}dz_a(p)
    =&
    -\frac{1}{8\nu_a}
    \left(  2 \sum_{n=1}^\infty nt_{an}v_{an} +t_{a0}^2\right).
 \end{split}
\end{equation}
Therefore, the free energy is given explicitly by
\begin{equation*}
 F =
 -\frac{1}{8}\left(\frac{t_{00}^2}{\nu_0}
 +\sum_{a=1}^{M}\frac{t_{a0}^2}{\nu_a}\right)
 + \frac{1}{2}\sum_{a=1}^M t_{a0}v_{a0}
 +\frac{1}{2}\sum_{\alpha=0}^M
             \sum_{n=1}^{\infty}
             \left(1- \frac{n}{2\nu_{\alpha}}\right)t_{\alpha n}v_{\alpha n}.
\end{equation*}

\end{document}